\documentclass{amsart}

\usepackage{a4wide}
\usepackage{amssymb,amsmath,graphicx,subfig}
\usepackage{tikz}
\usepackage{mathptmx}

\newcommand{\eps}{\epsilon}

\newcommand{\bs}{\bar{\sigma}}
\newcommand{\bmv}{\bar{\mu}_v}
\newcommand{\bW}{\mathbf{W}}

\renewcommand{\O}{\mathcal{O}}
\newcommand{\R}{\mathbb{R}}
\newcommand{\rnn}{\R^3_{>0}}
\newcommand{\e}{\mathrm{e}}
\newcommand{\bd}{\bar{\delta}}
\newcommand{\bmh}{\bar{\mu}_h}
\newcommand{\bg}{\bar{\gamma}}
\newcommand{\bR}{\bar{R}_0}
\renewcommand{\S}{\mathcal{S}}

\newcommand{\rd}{\mathrm{d}}
\newcommand{\bX}{\bar{X}}
\newcommand{\bY}{\bar{Y}}
\newcommand{\bZ}{\bar{Z}}

\newcommand{\hZ}{\hat{Z}}
\newcommand{\hbW}{\hat{\mathbf{W}}}
\newcommand{\bQ}{\mathbf{Q}}

\newcommand{\I}{I}
\newcommand{\fX}{\mathfrak{X}}
\newcommand{\hfX}{\hat{\mathfrak{X}}}

\newtheorem{thm}{Theorem}
\newtheorem{rmk}{Remark}
\newtheorem{lem}{Lemma}
\newtheorem{prop}{Proposition}

\graphicspath{{./Figures/}}

\def\squareforqed{\hbox{\rlap{$\sqcap$}$\sqcup$}}
\def\qed{\ifmmode\else\unskip\quad\fi\squareforqed}
\def\smartqed{\def\qed{\ifmmode\squareforqed\else{\unskip\nobreak\hfil
\penalty50\hskip1em\null\nobreak\hfil\squareforqed
\parfillskip=0pt\finalhyphendemerits=0\endgraf}\fi}}

\smartqed

\begin{document}
\title{Multiscale Analysis for a Vector-Borne Epidemic Model}
\author{Max O. Souza}
\address{%
Departamento de Matem\'atica Aplicada, Universidade Federal
Fluminense, R. M\'ario Santos Braga, s/n, 22240-920, Niter\'oi, RJ, Brasil.
}
\email{msouza@mat.uff.br}

\date{\today}

\thanks{The author acknowledges many useful discussions held in the Dengue Modeling initiative developed at the CMA/FGV, where part of this work was performed. The author also acknowledges the workshops and support of PRONEX Dengue under  CNPQ grant \# 550030/2010-7. The author is partially supported by CNPq grant \# 309616/2009-3 and FAPERJ grant \# 110.174/2009.}

 \begin{abstract}
Traditional studies about  disease  dynamics have focused on global stability issues, due to their epidemiological importance. We study a classical SIR-SI model for arboviruses in two different directions: we begin by describing an alternative proof of previously known global stability results by using only a Lyapunov approach. In the sequel,  we  take a different view and we  argue that vectors and hosts can have very distinctive intrinsic time-scales, and that such distinctiveness extends to the disease dynamics.   Under these hypothesis, we show that  two asymptotic regimes naturally appear: the fast host dynamics and the fast vector dynamics. The former regime yields, at leading order,  a SIR model for the hosts, but  with a rational incidence rate. In this case, the vector disappears from the model, and the dynamics is similar to a directly contagious disease. The latter yields a SI model for the vectors, with the hosts disappearing from the model. Numerical results show the performance of the approximation, and a rigorous proof validates the reduced models.


\end{abstract}

\maketitle

\section{Introduction}

Vector-borne diseases in general, and arboviruses in particular, are a contemporary major challenge for epidemiologists, public health officers to name a few.  Indeed, while  in the turning  from the nineteenth  to the twentieth century witnessed a discovering caused by arboviruses as for instance dengue and malaria, the turning from the twentieth to twentieth first century is witnessing a sustained  emergence of these diseases around the globe. Presently, Dengue is is a leading cause of serious illness and death among children in some Asian and Latin American countries---\cite{who:dengue}--- and Malaria pathogens are acquiring resistance to the first line of treatment in South Asia \cite{Breman:2012fu}. Additionally, the West Nile virus is now endemic in Africa, Asia, Oceania and it now established North America  \cite{Petersen:2012ye}, while Chinkunguya disease which has its origins in Africa has now progressed into Southern Asia and Oceania---see \cite{Powers:2000kl,Powers:2010qa}---and now there are documented cases in Europe \cite{cdc:chikungunya}. Such an emergence seems to be mainly caused by the spreading of some the associated vectors. Thus, \textit{Aedes aegypti} and \textit{Aedes albopictus} have experienced a major increase in spreading in recent decades \cite{Lambrechts:2010kh}. Moreover, some of the diseases evolve from being benign   into more lethal forms. This is the case, for instance, of the Dengue  Haemorrhagic Fever which has become significantly more prevalent in recent years cf. \cite{Gubler:1998mi}. Thus even for  innocuous diseases as Chinkunguya disease, there are concerns of further antigenic evolution---see for instance \cite{Rajapakse:2010fu}.  In addition, there are also arboviruses that attack only animals as, for instance  heartworm in dogs, which is transmitted by \textit{Aedes albopictus} \cite{Gratz:2004}. Finally, we also have a number of diseases caused by tick \cite{Dantas-Torres:2012qo}.

Traditional modelling in epidemiology focuses on global stability of equilibria, since this characterises if a disease will become endemic and this is a major concern for public health officers. This view dates back to the original work by Ross on Malaria---cf. \cite{Ross1911} and \cite{Macdonald1957}, where the concept of a basic reproductive number ($R_0$) was introduced and became a modelling paradigm---see \cite{Smith:2012} for a very recent review on the works by Ross and Macdonald from a medical modelling point of view. In a fairly large class of models, we can define $R_0$ unambiguously and it can be shown that if $R_0<1$ the disease is extinct while if $R_0>1$ it becomes endemic---see \cite{Diekmann:2000}. In the former case this usually means that the disease-free state is locally asymptotically stable, while in the latter case may indicates an analogous situation for a disease-present state or simply that the disease is prevalent---in the sense used in dynamical systems parlance. The literature on mathematical epidemiology is too vast, and we limit ourselves to some references:
\cite{britton2003,Murray:2002} for textbook introductions and  \cite{Diekmann:2000,Anderson:1995,Bailey75,Dietz75} for both  contemporary and classical research monographs; see also the reviews in  \cite{Grassly:2008kx,Nishiura2006,Andraud:2012cq,Koella1991,Mandal:2011ve,Luz:2010}.

In the specific case of arboviruses, the dynamics and ecology  of both the vector and host turn out to be important. We refer the reader to \cite{Lord:2007ys,Scott:2000dq,Muir:1998qf,Scott:2000cr,Barrera2006} for general information in describing the  ecology of mosquitoes in general and of the \textit{Aedes  aegypti} and \textit{Aedes albopictus} in particular.   From a modelling perspective, a number of questions have been investigated as the possibility of vertical transmission  for Dengue in \cite{Adams:2010pi}, effect of temperature \cite{Yang:2009tg,Yang:2009kl},  insecticide resistance \cite{Luz:2009rw}, and the connection between imported cases and co-infections by different serotypes \cite{Aguiar:2011ff}.

In what follows, we take a different view and look at different features of  the dynamics of such arboviruses diseases. More specifically, we argue that  the vector and hosts can have very different time-scales, and upon this assumption, show how  to use classical ideas of asymptotic analysis to derive new models from old ones, and how these simplified models might contribute to the understanding of such dynamics.

\subsection{Time-scales in vector-borne diseases}
\label{ssec:intro_ts}

Vector-borne diseases are different from direct contagious ones, since there is indirect transmission from host to the vector and vice-versa. Their dynamical behaviour will depend both on the dynamics of the vector, of the host and on their interaction. In order to organise the discussion, we  take as a basic framework  the simplest, and probably the most natural, model from the point of view of mass-action epidemiological modelling: the coupling of a SIR model for the host and a SI model for the vector that was first developed by \cite{Bailey75,Dietz75}, and  is given schematically in Figure~\ref{mfig}.
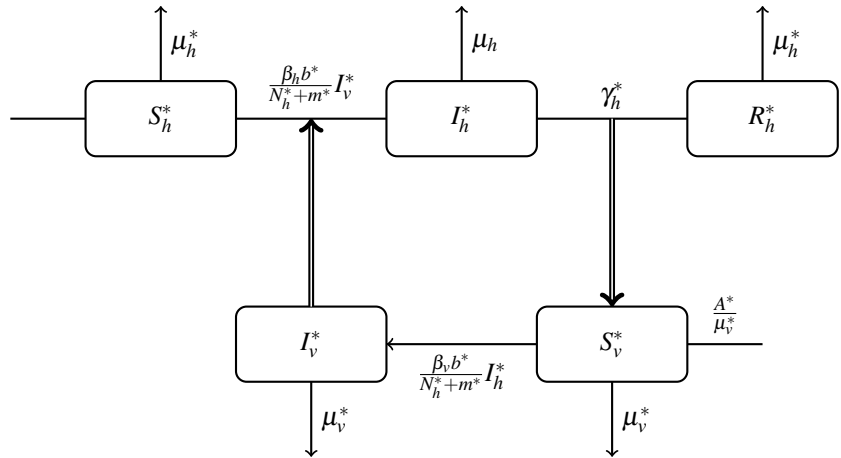
\begin{figure}[htb]
 \begin{center}
  \begin{tikzpicture}[thick,rounded corners]
 \draw  (0,0) rectangle (2,1);
\draw (1,0.5) node {$S^*_h$};
\draw (4,0) rectangle (6,1);
\draw (5,0.5) node {$I^*_h$};
\draw (8,0) rectangle (10,1);
\draw (9,0.5) node {$R^*_h$};
\draw (2,-2) rectangle (4,-3);
\draw (3,-2.5) node {$I^*_v$};
\draw (6,-2) rectangle (8,-3);
\draw (7,-2.5) node {$S^*_v$};
\draw (-1,0.5)--(0,0.5);
\draw (2,0.5)--(4,0.5);
\draw (3,0.5) node [above] {$\frac{\beta_hb^*}{N_h^*+m^*}I_v^*$};
\draw (6,0.5)--(8,0.5);
\draw (7,0.5) node [above] {$\gamma_h^*$};
\draw (9,-2.5)--(8,-2.5);
\draw (8.5,-2.5) node [above] {$\frac{A^*}{\mu_v^*}$};
\draw [->] (6,-2.5)--(4,-2.5);
\draw (5,-2.5) node [below] {$\frac{\beta_vb^*}{N_h^*+m^*}I_h^*$};
\draw [double,->] (3,-2)--(3,0.5);
\draw [double,->] (7,0.5)--(7,-2);
\draw [->] (1,1)--(1,2);
\draw (1,1.5) node [right] {$\mu_h^*$};
\draw [->] (5,1)--(5,2);
\draw (5,1.5) node [right] {$\mu_h$};
\draw  [->](9,1)--(9,2);
\draw (9,1.5) node [right] {$\mu_h^*$};
\draw [->] (3,-3)--(3,-4);
\draw (3,-3.5) node [right] {$\mu_v^*$};
\draw [->] (7,-3)--(7,-4);
\draw (7,-3.5) node [right] {$\mu_v^*$};
  \end{tikzpicture}
\caption{Compartmental description of the arbovirus model by \cite{Bailey75,Dietz75}. \label{mfig}}
 \end{center}
\end{figure}
\begin{table}[t]
\begin{center}
 \begin{tabular}{l||l}
 Parameter&Meaning\\\hline
$N^*_h$ and $N^*_v$ & Number of hosts and vectors;\\\hline
$\mu^*_h$ and $\mu^*_v$& birth rate for hosts and vectors; \\\hline
$\beta^*_h$ and $\beta^*_v$& probability of a host being infected by a vector and vice-versa; \\\hline
$b^*$& Biting rate;\\\hline
$\gamma^*$& removal rate;    \\\hline
$m^*$& Number of alternative blood sources; \\\hline
$A^*$& Vector recruitment rate.\\
\end{tabular}
\end{center}
\caption{Description of parameters  meaning in the compartmental model depicted in Figure~\ref{mfig}.\label{tbl:param}}
\end{table}

The meaning of the model parameters that appear in  Figure~\ref{mfig} is given in Table~\ref{tbl:param}. Here starred variables indicate dimensional quantities.  The variables $S$, $I$ and $R$ have the usual epidemiological meaning, with the subscript indicating if they refer to the hosts or to the vectors. In Figure~\ref{mfig}, we can view each arrow single arrow as a clock that determines when there is a compartment transition. The typical time-scale of these clocks will depend on the particular value of the parameters, and we  give some indication below of the possible scales that one might expect to observe. As in any description of parameters of epidemiological models, we should stress that their variability can be very large, and extremely dependent on particular factors as the pathogen itself, the vector, and the geography among many others. See \cite{Luz:2003ly} and \cite{Desenclos:2011tw} for two different discussions about these variability.

For the host compartments, that are four clocks: $N^*_h\mu^*_h$, $\mu_h^*$, $\beta_hb^*/(N_h^*+m^*)I^*_v$ and $\gamma^*$. The choice of the first one is a standard simplification, and implies that the population size remains constant. While this is a reasonable modelling assumption in many cases,  there might be exceptions when one is interested in a long time span, or when the mortality is high due to the deaths caused by the disease \cite{Ngwa:2000}. The second one is the host death rate, and it is related to average host life expectancy. Naturally, it depends both on  the host and on  the particular geographical region, but for humans we  have $1/mu_h\geq50$ years, unless  in case of diseases with severe mortality, as it can be the case of large epidemics of yellow fever due to the hight mortality of the severe cases \cite{Tomori:2004fc}.  The third clock  controls the infection events, and as we shall see in Section~\ref{sec:prelim}, it can be faster or slower than the analogous clock for the vectors. The fourth clock is the removal rate. Its values depends on the particular disease and on the particular host. For instance, for Dengue $1/\gamma$ can vary from 1 to eight days, but there are indications that it can be larger in patients with Dengue Haemorrhagic Fever \cite{Chastel:2012gb}. In Malaria it can varies from $5$ to $10$ days in treated adults, but it can be as large as $48$ days in young children, and $120$ days in non-treated patients \cite{Okell:2008bh}.

For the vector compartments, there are three clocks: $A^*/\mu_v^*$,  like its host analogue, is  chosen with $A^*$ constant, so that the vector population size also remains constant.  The second clock is the vector death rate. In laboratory, one can have female mosquitoes living about $20$ days, cf. \cite{Styer:2007nx}, while in the wild this expected lifespan can be as small as $2$ days \cite{David:2009}. These values are also affected by variables as temperature, rainfall among others \cite{Chen:2012jl}. The last clock describes the frequency at which mosquitoes get infected. It has a large variance depend on the species---e.g.. \cite{Aldemir:2010zt}. For \textit{Aedes aegypti},  it can be larger due to the fact that  its feeding is easily disturbed, and hence it might need to bite many hosts in order to complete its feeding \cite{Scott:1993rq}.

The discussion above suggest that there might be a number of situations, where one can have all the vector clocks faster than the host clocks; we call such a situation the \textit{Fast Vector Dynamics}. The dual regime, the \textit{Fast Host Dynamics}, where one has the host clocks faster than the vector clocks seems less likely but, as observed above, situations like large epidemics of yellow fever might be a possible scenario where such model is relevant.  A more precise identification of these different regimes in terms of the underlying parameters is  deferred to Section~\ref{sec:prelim}.

\subsection{Outline}
In Section~\ref{sec:prelim}, we describe the basic model studied and review some of its properties. An alternative presentation of  global stability results, using only a Lyapunov approach, is given in Section~\ref{sec:global}.
In Section~\ref{sec:fvd}, we study the \textit{Fast Vector Dynamics} limit. This regime yields a reduced SIR system with  modified non-linear incidence rates In addition, it also yields  a transition layer corrector for which an explicit solution can be written, and an explicit representation of the so-called slow manifold can be found.   We also present a number of numerical illustrations of the approximation together with a theorem that guarantees that such an approximation is uniform for all time, with the proof given in \ref{app:asymp_proof}. 
The dual limit, the \textit{Fast Host Dynamics}, yields a reduced SI system and it seems to be somewhat less interesting from a biological point of view; it is presented briefly in Section~\ref{sec:fhd}; the numerical results are similar to the ones obtained for the \textit{Fast Vector Dynamics} and thus are omitted. Nevertheless, we do illustrate the reduction to a one-dimensional slow manifold. Section~\ref{sec:conclusion} draws some concluding remarks.

\section{Preliminaries}
\label{sec:prelim}

\subsection{Review of the model and non-dimensionalisation}

The compartmental model shown in Figure~\ref{mfig} is described by the following system:
\begin{equation}
\left\{
\begin{array}{rcl}
\dot{S_h^*}&=&\mu_h^*(N_h^*-S_h^*)-\frac{\beta_h^*b^*}{N_h^*+m^*}S_h^*I_v^*\\
\dot{I_h^*}&=&\frac{\beta_h^*b^*}{N_h^*+m^*}S_h^*I_v^* - (\mu_h^*+\gamma^*)I_h^*\\
\dot{R_h^*}&=&\gamma^* I_h-\mu_h^*R_h^*\\
\dot{S_v^*}&=&A^*-\mu_v^*S_v^* - \frac{\beta_v^*b^*}{N_h^*+m^*}S_v^*I_h^*\\
\dot{I_v^*}&=&\frac{\beta_v^*b^*}{N_h^*+m^*}S_v^*I_h^*-\mu_v^*I_v^*
\end{array}
\right.
\label{mod:bas}
\end{equation}
System~\eqref{mod:bas}  has been comprehensively studied by \cite{Esteva1998}, and extensively used for studies of Dengue as described, for instance, in \cite{Nishiura2006}. 

We non-dimensionalise  system~\eqref{mod:bas} by letting
\[
 (S_h^*,I_h^*,R_h^*)=N_h^*(S_h,I_h,R_h)\quad\text{and}\quad
(S_v^*,I_v^*)=\frac{A^*}{\mu_v^*}(S_v,I_v).
\]
Also
\[
 t^*=\left(\Delta^*\right)^{-1} t,
\]
where $\left(\Delta^*\right)^{-1}$ is, for the time being, an arbitrary time-scale. We immediately obtain the new system:
\begin{displaymath}
\left\{
\begin{array}{rcl}
\dot{S_h}&=&\mu_h(1-S_h)-\delta S_hI_v\\
\dot{I_h}&=&\delta S_hI_v - (\mu_h+\gamma)I_h\\
\dot{R_h}&=&\gamma I_h-\mu_h R_h\\
\dot{S_v}&=&\mu_v(1-S_v) - \sigma S_vI_h\\
\dot{I_v}&=&\sigma S_vI_h-\mu_v I_v
\end{array}
\right.
\end{displaymath}
where
\[
 \delta= \frac{\beta_h^*b^*A^*}{\mu_v^*\Delta^*(N_h^*+m^*)}\quad
 \text{and}\quad
\sigma=\frac{\beta_v^*b^*N_h^*}{\Delta^*(N_h^*+m^*)}
\]
Also 
\[
 \gamma=\frac{\gamma^*}{\Delta^*},\quad
 \mu_h=\frac{\mu_h^*}{\Delta^*}\quad\text{and}\quad
\mu_v=\frac{\mu_v^*}{\Delta^*}.
\]
As observed in \cite{Esteva1998}, the non-negative orthant of $\R^5$ is invariant by the flow of \eqref{mod:bas}, and it is conservative, i.e., if  the initial values of the host compartments sum to $N_h^*$ and the initial values of the vector compartments sum to $A^*/mu^*_v$, then this holds for all time.

In view of this observation,  if the initial conditions for the host fractions add to one, with the same being true for the vector fractions initial conditions, then this is preserved by the evolution. Therefore, we work with the simplified, but equivalent, model below:
\begin{equation}
 \left\{ 
\begin{array}{rcl}
 \dot{X} &=& \mu_h(1-X)-\delta XZ\\
\dot{Y}  &=& \delta XZ - (\mu_h+\gamma)Y\\
\dot{Z}  &=&  \sigma(1-Z)Y-\mu_v Z
\end{array}
\right.
\label{red:sys}
\end{equation}
where
\[
X=S_h,\quad Y=I_h
\quad\text{and}\quad
Z=I_v.
\]

System \eqref{red:sys} has  the following  two equilibrium points: 
\begin{enumerate}
 \item The disease free equilibrium: $X^*=1$, $Y^*=Z^*=0$.
\item The endemic equilibrium:
\[
 X^*=\frac{1}{R_0}\frac{1+R_0D_0}{1+D_0},\quad Y^*=\frac{\mu_h}{\mu_h+\gamma}(1-X^*),\quad Z^*=D_0\frac{(1-X^*)}{X^*},
\]
\[
 R_0=\frac{\sigma\delta}{\mu_v(\mu_h+\gamma)}\quad\text{and}\quad D_0=\frac{\mu_h}{\delta}.
\]
\end{enumerate}

The dynamics of system~\eqref{red:sys} was studied by \cite{Esteva1998} who showed, using the theory of monotone dynamical systems,  the global stability of the equilibria: the disease free equilibrium for $R_0\leq1$,  and the endemic equilibrium for $R_0>1$. They also showed that, for sufficiently small  $\mu_h$, the approach to the endemic equilibrium is oscillatory.  We shall return to the question of  global stability in Section~\ref{sec:global}.

\subsection{Scalings}

We now discuss what parameters values in system \eqref{red:sys} may lead to fast vector dynamics.   We refer the reader to the introduction for a more biological discussion and for more general references.

We use $\Delta^*=\gamma^*/2$, so that $\left(\Delta^*\right)^{-1}$ is twice the typical time that host remains infectious once it has acquired the disease. Thus, we always have $\gamma=1/2$.  We also take $m^*=0$ for simplicity, and $\beta_h=\beta_v=0.5$.For the host death rate, we  assume an average life expectancy of 60 years; thus we use $\mu_h^*=0.0000463$ $\mathrm{days}^{-1}$.

As discussed above, the infectious time can vary depending on the disease, age, and if there is treatment available. The biting rate will depend on the particular vector, and of the are being modelled. Also, a typical value of $A^*/\mu_v^*$ is dependent on the infestation level. Therefore, we do not make any a priori-hypothesis on those values, and instead we  show the variation of the non-dimensional parameters. Also, since $\mu_h^*$ is quite small, $\mu_h$ is also  small, and we concentrate on $\delta$, $\mu_v$ e $\sigma$.

We begin with $\mu_v$,  and show contour plots as a function of $(\Delta^*)^{-1}$ and of the expected lifespan of the mosquito in Figure~\ref{fig:muv_scaling}.  
\begin{figure}[htbp]
\begin{center}
\includegraphics[scale=0.4]{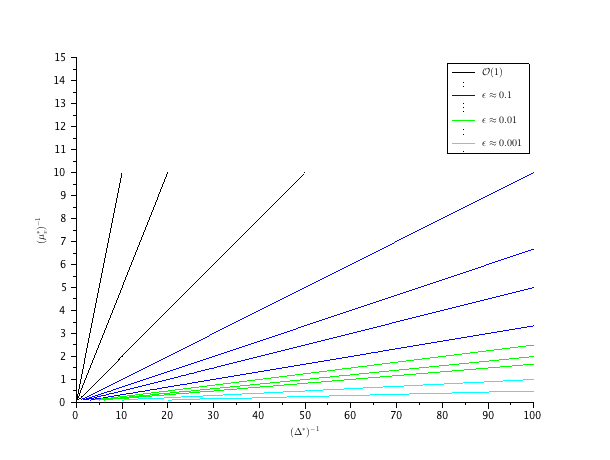}
\caption{Scaling  for $\mu_v$ depending  on the average lifespan of the vector, and on twice the average infectious time of the host. Although the level curves are just lines, the graph helps to assess the order size of the parameters. A level curve with value $L$ was labelled with a certain value of $\epsilon$, if $\log(L\epsilon)\in [-1,0.8]$. This is equivalent to say that $L=a\epsilon^{-1}$, $a\in [a_0,a_1]$, with $a_0\approx0.2$ and $a_1\approx2$. Based on the discussion in Section~\ref{ssec:intro_ts} and on the values on the level curves, we observe that the range $0.01\leq\epsilon\leq0.1$ seems to match a number of possible scenarios in epidemiological modelling.\label{fig:muv_scaling} }
\end{center}
\end{figure}

With the assumed parameters, we have that
\[
\sigma=b^*/2\Delta^*
\quad\text{and}\quad 
\delta=\sigma \I.
\] 
Thus, in Figure~\ref{fig:sigma_delta_scaling } we plot the level curves of these expressions. See the corresponding caption for more details, and the caption corresponding to Figure~\ref{fig:muv_scaling} for an explanation about the labelling of the level curves.

\begin{figure}[htbp]
\begin{center}
\subfloat{\includegraphics[scale=0.345]{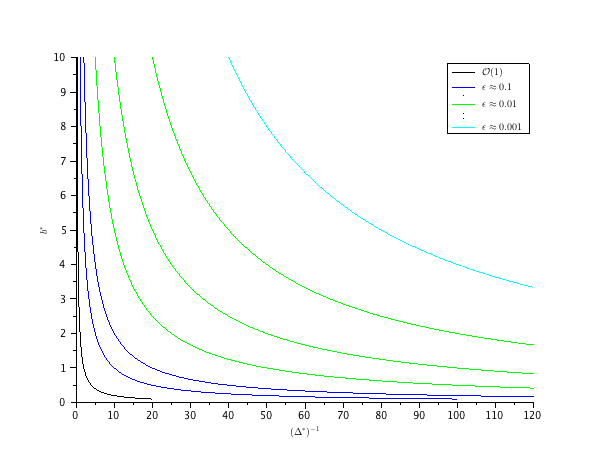}}
\subfloat{\includegraphics[scale=0.345]{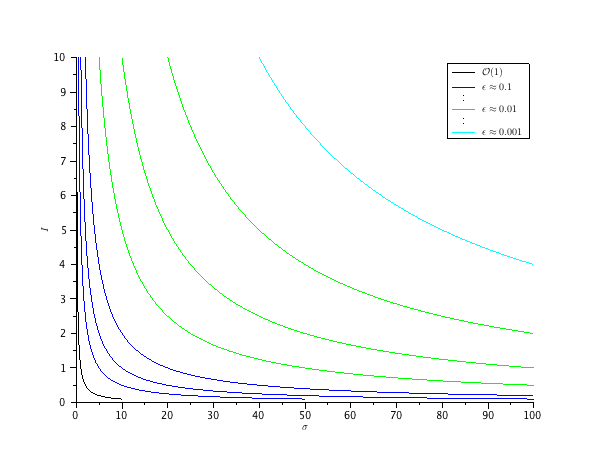}}
\caption{Possible scalings  for $\sigma$ and $\delta$.For the values of $b^*$ that obtained in the literature, the level curves of $\sigma$ are naturally around $\sigma\sim(0.01)^{-1}$.  Notice that  we  can easily have $\delta$  of the same order of magnitude than $\sigma$, in the case of reasonably large infestations. However,   for more moderate cases with $ \I\approx 0.1$,  we typycally have $\delta=\O(1)$. We refer to the discussion in Section~\ref{ssec:intro_ts} for a discussion about these values. \label{fig:sigma_delta_scaling }}
\end{center}
\end{figure}

The pictures in Figures~\ref{fig:muv_scaling} and \ref{fig:sigma_delta_scaling } suggest both that the proposed regimes are effective, but also that might be a number of other regimes that need to be studied in addition. Additionally, we point out  that $\epsilon=0.01$, seems to be a plausible scaling for Dengue in urban centres---the reader is referred  to  the discussion in Section~\ref{ssec:intro_ts}.
\newpage

\section{Global stability analysis}
\label{sec:global}

Traditionally, the main focus on study about epidemiological systems as~\eqref{mod:bas} is on global stability issues. 
In particular a proof of global stability for system~\eqref{mod:bas} using the theory of competitive systems can be found in \cite{Esteva1998}. More recently, \cite{Cai2009} has studied a similar system but with a saturated bilinear incidence. In  particular, \cite{Cai2009} proved global stability by using a Lyapunov function for the disease-free equilibrium, and using the theory of competitive systems for the endemic equilibrium, whenever it exists. \cite{Tewa2009} has studied the same system, and provide a Lyapunov proof for the disease-free. The proof presented by \cite{Tewa2009} for the endemic equilibrium seems to applies only  if, in our notation, one has $X^*=1$. But in this case, we must necessarily have $R_0=1$.  Also a proof using a Lyapunov approach for the disease free equilibrium, and another proof using the theory of competitive systems for the endemic equilibrium can be found in \cite{Yang2010-04-01}.

In what follows,  we provide an alternative proof for the global stability of \eqref{mod:bas}  using only Lyapunov functions. 

\begin{thm}
\label{thm:3d_gs}
 Let $R_0$ be defined as above. Then for $R_0\leq 1$ the disease-free equilibrium is globally asymptotic stable, while for $R_0>1$ the endemic equilibrium is globally asymptotic stable.
\end{thm}
\begin{proof}
 Although the proof for $R_0\leq1$ is available in the literature, as discussed above, we provide it here both for the sake of completeness of presentation, and because the proof given here is somewhat simpler that the available ones. 
 
Suppose $R_0\leq 1$ and consider the following Lyapunov function
\[
 V(X,Y,Z)=X-\log X + Y +\frac{\delta}{\mu_v}Z.
\]
Then
\[
 \dot{V}=\dot{X}\left(1-\frac{1}{X}\right)+\dot{Y}+ \dot{Z}
\]
On substituting, we obtain:
\[
 \dot{V}=-\mu_h\frac{(1-X)^2}{X}-(\mu_h+\gamma)(1-R_0)Y-R_0ZY,
\]
which is negative for $0\leq R_0\leq1$, and $X,Y,Z$ in $\rnn$. 

For $R_0>1$, let
\[
 V(X,Y,Z)=X-X^*\log \frac{X}{X^*} + Y -Y^*\log\frac{Y}{Y^*}+\frac{\delta X^*}{\mu_v}\left(Z-Z^*\log\frac{Z}{Z^*}\right).
\]
Then we have
\begin{align*}
 \dot{V}=& \dot{X}\left(1-\frac{X^*}{X}\right) + \dot{Y}\left(1-\frac{Y^*}{Y}\right) + \frac{\delta X^*}{\mu_v}\dot{Z}\left(1-\frac{Z^*}{Z}\right)\\
=& \mu_h\left[1+X^* -X -\frac{X^*}{X}\right] + (\mu_h+\gamma)Y^* +\frac{\delta X^*}{\mu_v}Z^* + \frac{\sigma\delta X^*}{\mu_v}(1-Z)Y+\\
&\qquad -\delta\frac{XZY^*}{Y}-(\mu_h+\gamma)Y-\frac{\sigma\delta X^*Z^*}{\mu_v}(1-Z)\frac{Y}{Z}\\
\end{align*}
On using that
\[
(\mu_h+\gamma)Y^*=\mu_h(1-X^*)
\quad\text{and}\quad
\frac{\delta X^*}{\mu_v}Z^* =\mu_h(1-X^*),
\]
we arrive at
\begin{align*}
\dot{V}=&\mu_h\left[3-X^*-X-\frac{X^*}{X}\right]+(\mu_h+\gamma)\left(R_0X^*(1-Z^*)-1\right)Y+\\
&\qquad +\frac{\sigma\delta X^*}{\mu_v}Y\left[2Z^*-Z-\frac{Z^*}{Z}\right]-\frac{\delta XZY^*}{Y}.\\
\end{align*}
Since
\[
 R_0(1-Z^*)=R_0\frac{1+D_0}{1+R_0D_0}=\frac{1}{X^*},
\]
we are left with
\[
 \dot{V}=\mu_h\left[3-X^*-X-\frac{X^*}{X}\right]+\frac{\sigma\delta X^*}{\mu_v}Y\left[2Z^*-Z-\frac{Z^*}{Z}\right]-\frac{\delta XZY^*}{Y}.
\]
Writing
\[
 X+\frac{X^*}{X}=X+\frac{(X^*)^2}{X}+\frac{X^*(1-X^*)}{X},
\]
using that
\[
 X+\frac{(X^*)^2}{X}\geq 2X^*,
\]
and on noticing that a similar calculation holds for the second bracket, we find:
\begin{equation}
 \dot{V}\leq 3\mu_h(1-X^*) -\mu_h(1-X^*)\frac{X^*}{X}-
\frac{\sigma\delta X^*}{\mu_v}\frac{Z^*(1-Z^*)}{Z}Y-\frac{\delta XZY^*}{Y}.
\label{eqn:part_ineq}
\end{equation}
%

Let us write
\[
R=-\mu_h(1-X^*)\frac{X^*}{X}-
\frac{\sigma\delta X^*}{\mu_v}\frac{Z^*(1-Z^*)}{Z}Y-\frac{\delta XZY^*}{Y}.
\]
By the Arithmetic-Geometric inequality, we have
\begin{align*}
R\leq& -3\left[\mu_h(1-X^*)X^*\frac{\sigma\delta}{\mu_v}X^*Z^*(1-Z^*)\delta Y^*\right]^{1/3}\\
&=-3\left[\mu_h^2(1-X^*)^3D_0X^*R_0(1-Z^*)\delta\right]^{1/3}\\
&=-3(1-X^*)\left[\mu_h^2D_0\delta\right]^{1/3}\\
&=-3(1-X^*)\mu_h.
\end{align*}
Hence, we have that 
\[
\dot{V}\leq0,
\]
in $\rnn$.

In order to show that the $V$ is a strict Lyapunov function, we recall that the inequality in \eqref{eqn:part_ineq} is strict unless $X=X^*$ and $Z=Z^*$. However, in this case we have that
\begin{align*}
R&=-\mu_h(1-X^*)-\frac{\sigma\delta X^*}{\mu_v}(1-Z^*)Y-\frac{\delta X^* Z^* Y^*}{Y}\\
&=-\mu_h(1-X^*)-(\mu_h+\gamma)Y-\delta D_0(1-X^*)\frac{Y^*}{Y}\\
&=-\mu_h(1-X^*)\left[1+\frac{Y}{Y^*}+\frac{Y^*}{Y}\right]
\end{align*}
Since
\[
\frac{Y}{Y^*}+\frac{Y^*}{Y}\geq2,
\]
with equality only when $Y=Y^*$, we conclude that $R<-\mu_h(1-X^*)$, except on $(X,Y,Z)=(X^*,Y^*,Z^*)$ where we have equality.
Therefore, we have $\dot{V}<0$ in  $\rnn$, except in the endemic equilibrium. 
\qed
\end{proof}

\section{The fast vector dynamics}
\label{sec:fvd}

As discussed above, we want to  describe the dynamics of system~\eqref{red:sys}, when we have:
\[
 \sigma=\bs\eps^{-1}\quad\text{and}\quad\mu_v=\bmv\eps^{-1},\quad 0<\epsilon\ll1,
\]
while all other parameters are of order one. 

Direct substitution in \eqref{red:sys} lead us to the following initial value problem:
\begin{equation}
 \left\{ 
\begin{array}{rcl}
 \dot{X} &=& \mu_h(1-X)-\delta XZ\\
\dot{Y}  &=& \delta XZ - (\mu_h+\gamma)Y\\
\eps\dot{Z}  &=&  \bs(1-Z)Y-\bmv Z
\end{array}
\right.
\label{fvd:eps}
\end{equation}
subject to the initial condition 
\[
 X(0)=X_0,\quad Y(0)=Y_0\quad\text{and}\quad Z(0)=Z_0.
\]
On a formal basis, since we already know that the dynamics of \eqref{red:sys} always converges to an equilibrium, we expect that the right hand side of the last equation balances out, leaving $\dot{Z}\approx0$, i.e, that the vector population is nearly in equilibrium. Under these hypothesis, we obtain the following system:
\begin{equation}
 \left\{ 
\begin{array}{rcl}
 \dot{X} &=& \mu_h(1-X)-\delta \frac{\sigma XY}{\sigma Y+\mu_v}\\
\dot{Y}  &=& \delta \frac{\sigma XY}{\sigma Y+\mu_v} - (\mu_h+\gamma)Y\\
\dot{Z}  &=&  0
\end{array}
\right.
\label{fvd:heu}
\end{equation}
Notice that system~\eqref{fvd:heu} can be seen as a SIR system with a modified, rational, incidence rate. While the above derivation  is heuristic, we now show that it can obtained from a consistent multiscale asymptotic expansion and, moreover, that such an expansion can be rigorously justified.

\subsection{Asymptotic expansion}
\label{subsec:asymp_exp}
Let 
\[
 \eps\tau=t.
\]
Then, we seek a composite  expansion of the form
\begin{align*}
 X&=X^0(t) +  \O(\eps)\\
Y&=Y^0(t) + \O(\eps)\quad\text{and}\\
Z&=Z^0(t)+\hat{Z}^0(\tau)+ \O(\eps),
\end{align*}
where
\[
 \lim_{\tau\to\infty}\hat{Z}^0(\tau)=0.
\]

On substituting the proposed expansion in \eqref{fvd:eps}, we obtain  to  leading order  the following differential-algebraic system:
\begin{align*}
 X^0_t=&\mu_h(1-X^0)-\delta X^0Z^0\\
Y^0_t=&\delta X^0Z^0-(\mu_h+\gamma)Y^0\\
0=&\bs(1-Z^0)Y^0-\bmv Z^0.
\end{align*}
This yields
\begin{equation}
 Z^0(t)=\frac{\bs Y^0(t)}{\bs Y^0(t)+\bmv}.
 \label{eqn:slaving}
\end{equation}
and hence, we obtain the system:
\begin{equation}
\left\{ 
\begin{array}{rcl}
 X^0_t&=&\mu_h(1-X^0)-\delta \frac{\bs X^0Y^0}{\bs Y^0+\bmv}\\
Y^0_t&=&\delta \frac{\bs X^0Y^0}{\bs Y^0+\bmv}-(\mu_h+\gamma)Y^0\\
\end{array}
\right.
\label{lo:fastvector}
\end{equation}
with initial condition $X^0(0)=X_0$ and $Y^0(0)=Y_0$.

Notice that, in general, we  have $Z^0(0)\not=Z_0$. Such a mismatch in the initial condition, should be corrected by $\hat{Z}^0$. Thus, $\hat{Z}^0$ should satisfy both $\hat{Z^0(0)=Z_0-Z^0(0)}$ (it adjusts for the correct initial condition) and $\hat{Z}^0(\tau)\to0$, as $\tau\to\infty$ (it has a local character). In order to solve for $\hat{Z}^0$, we first observe that
\[
 X^0(t)=X^0(\eps\tau)=X^0(0)+\eps\tau X^0_t(0)+\O(\eps^2),
\]
with similar expansions for $Y^0(t)$ and $Z^0(t)$.

Hence, we find that $\hZ^0$ satisfies
\[
 \hZ^0_\tau=-(\bs Y_0+\bmv)\hZ^0,
\]
i.e
\begin{equation}
 \hat{Z}^0(\tau)=\left(Z_0-\frac{\bs Y_0}{\bs Y_0+\bmv}\right)\e^{-(\bs Y_0 +\bmv)\tau}.
 \label{eqn:corrector}
\end{equation}

\subsection{Global stability analysis of the asymptotic system}

Before we can assert the quality of the approximation provided by  \eqref{eqn:slaving}, \eqref{lo:fastvector} and \eqref{eqn:corrector}, we need a better understanding of the dynamics of the reduced system.  We begin with  the following basic result:
\begin{prop}
Let
\[
 \S=\{(X^0,Y^0)\in\R^2\, |\, X^0+Y^0\leq 1,\quad X^0,Y^0\geq0\}.
\]
Then $\S$ is invariant by the flow of \eqref{lo:fastvector}. In particular, the corresponding solutions are global in time.
\end{prop}
\begin{proof}
 Since $Y^0=0$ is an invariant set, a solution with $Y^0\geq0$ at any time,  will remain this property. Also, when $X^0=0$, the flow points inside $\S$. Thus, for an initial condition in $\S$, we also have $X^0\geq0$. Finally, $(X^0+Y^0)_t\leq 0$.
 \qed
\end{proof}

System \eqref{lo:fastvector} has two equilibrium points in the positive quadrant:
\begin{enumerate}
 \item $X_*=1$, $Y_*=0$.
\item 
\[
 X_*=\frac{(\mu_h+\gamma)\bmv+\nu_h\bs}{(\mu_h+\delta)\bs}
\quad\text{and}\quad
Y_*=\frac{\delta\mu_n\bs-(\mu_h^2+\gamma\mu_h)\bmv}{\mu_h^2  + (\gamma + \delta) h + \delta \gamma)\bs}.
\]
\end{enumerate}
\begin{rmk}
 Notice that $X_*=X^*$ and $Y_*=Y^*$. Thus the equilibria of \eqref{lo:fastvector} correspond to the projections in the $XY$ plane of the equilibria of \eqref{red:sys}. Notice also that
\[
 Z_*=\frac{\bs Y_*}{\bs Y_*+\bmv}=Z^*.
\]
\end{rmk}
The next result shows that the dynamics of systems \eqref{fhd:eps} and \eqref{lo:fastvector} are qualitatively equivalent, in the sense that either both end up in the disease free equilibrium or in the endemic one.
\begin{thm}
\label{thm:pl_gs}
 Let 
\[
 \bR=\frac{\delta\bs}{\bmv(\nu_h+\gamma)}
\]
Then $\bR=R_0$, and for $R_0\leq1$ the disease free equilibrium is globally asymptotically stable. For $R_0>1$, the endemic equilibrium is globally asymptotically stable.
\end{thm}
\begin{proof}
The assertion about $R_0$ is clear. For the global stability, let
\[
 F(Y^0)=\frac{\bs Y^0 + \bmv}{Y^0}
\]
Then $F$ is a Dulac function for the system \eqref{lo:fastvector} in  compact subsets of $\S$ that do not intersect the $X^0$ axis, since we have:
\[
 \partial_{X^0}[\mu_h(1-X^0)F(Y^0)-\delta X^0]+\partial_{Y^0}[\delta X^0 - (\mu_h+\gamma)(\bs Y^0 + \mu_v)=
-\mu_hF(Y^0)-\delta -\bs(\mu_h+\gamma)<0.
\]
Thus, the system cannot have a closed orbit in the interior of $\S$. For $R_0\leq1$, the only equilibrium in $\S$ is $(1,0)$.  Thus  all orbits must converge to this equilibrium point. 

The linearisation of \eqref{lo:fastvector} is
\[
 \begin{pmatrix}
  \bar{X}^0\\\bar{Y}^0
 \end{pmatrix}_T
=
\begin{pmatrix}
 -\mu_h-\delta\frac{\bs Y^0}{\bs Y^0+\bmv}&-\bs\delta X^0\frac{\bmv}{(\bs Y^0+\bmv)^2}\\
\delta\frac{\bs Y^0}{\bs Y^0+\bmv}&\bs\delta X^0\frac{\bmv}{(\bs Y^0+\bmv)^2}-\mu_h-\gamma\\
\end{pmatrix}
\begin{pmatrix}
  \bar{X}^0\\\bar{Y}^0
 \end{pmatrix}
\]
For the disease free equilibrium, the eigenvalues of the Jacobian are $-\mu_h$ and $(\mu_h+\gamma)(R_0-1)$. Thus, the disease free equilibrium is a locally asymptotically stable node for $R_0 <1$, and a saddle for $R_0>1$. In the latter case, the unique orbit that approaches the disease free  equilibrium is easily shown to be the intersection of $\S$ with $Y^0=0$. Thus, all the other orbits must approach the endemic equilibrium.
\qed
 \end{proof}

\subsection{Asymptotic convergence and numerical results}

The asymptotic expansions derived in \ref{subsec:asymp_exp} can be shown to be indeed asymptotic. The ideas used here are similar to the ones used to formalise Kinetic Menton's theory---cf. \cite{OMalley1991} and references therein for instance. In particular, we have the following:
\begin{thm}
 Let $\fX^\epsilon(t)=(X(t),Y(t),Z(t))$ and $\fX^0(t,\tau)=(X^0(t),Y^0(t),Z^0(t)+\hat{Z}^0(\tau))$. Denote the uniform norm in $[0,\infty)$ by $\|.\|_\infty$. Then, for sufficient small $\eps>0$,  there exists a constant $C>0$, independent of $\eps$, such that
\[
 \|\fX-\fX^0\|_\infty\leq C\eps.
\]
Moreover, let $\\hfX^0(t)=(X^0(t),Y^0(t),Z^0(t))$. Then there are constants $C_1,C_2>0$  such that, for $t>C_1\epsilon$, we have
\[
\|\fX^\epsilon(t)-\hfX^0(t)\|\leq C_2\epsilon.
\]
\label{asymp_thm}
\end{thm}
The proof of Theorem~\ref{asymp_thm} is given in Appendix~\ref{app:asymp_proof}.

We now present some numerical illustrations. We first observe that the hallmark of the Fast Vector Dynamics is that
\begin{equation}
\frac{\bar{\sigma}(1-Z(t))Y(t)}{\bar{\mu}_vZ(t)}=1.
\label{eqn:ratio_fvd}
\end{equation}
In Figures~\ref{figratio_ic} and \ref{figratio_eps} we check performance of the approximation in terms of \eqref{eqn:ratio_fvd} for two different parameter sets: the first one has $\mu_h=0.4$, $\gamma=0.25$, $\delta=0.5$, $\bar{\mu}_v=0.2$, and $\bar{\sigma}=0.5$; the second set has $\mu_h=0.0005$, with the other parameters values equal to the first set.  Two verifications are performed: in Figure~\ref{figratio_ic}, we verify the performance of the approximation with a fixed small $\epsilon$, but with a number of different initial conditions; in Figure~\ref{figratio_eps}, we now fix an initial condition, but have a number of different values of $\epsilon$.

\begin{figure}[htbp]
\begin{center}
\subfloat[Parameter set 1]{\includegraphics[scale=0.345]{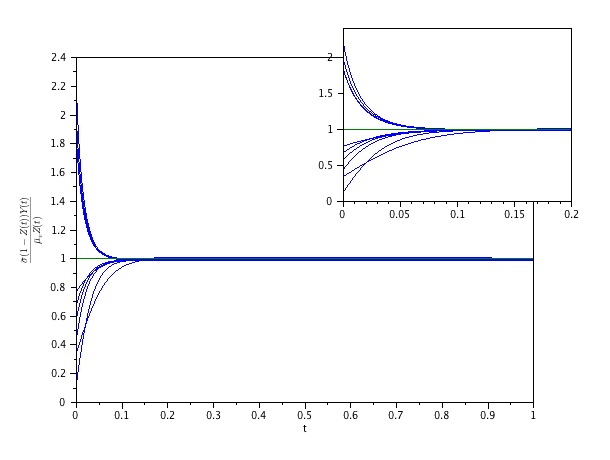}}
\subfloat[Parameter set 2]{\includegraphics[scale=0.345]{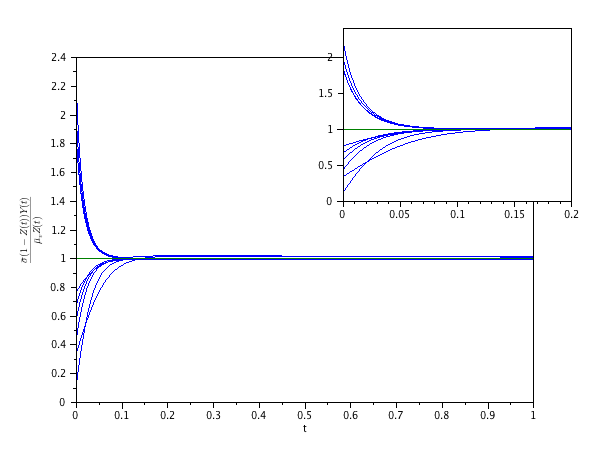}}
\caption{We show how the ratio in \eqref{eqn:ratio_fvd} is attained for the set of parameters described in the text. In these numerical experiments we have $\epsilon=0.01$, and a set of ten initial conditions that was randomically generated---but kept constant for the two figures. For both sets, the endemic equilibrium is globablly stable, but typically these equilibria are attained  arround time $t=10$ for the first parameter set, and about $t=700$ for the second set. Hence, convergence to \eqref{eqn:ratio_fvd} at earlier times, as shown in the graphs, is not a consequence of convergence to equilibrium.\label{figratio_ic}}
\end{center}
\end{figure}

\begin{figure}[htbp]
\begin{center}
\subfloat[Parameter set 1]{\includegraphics[scale=0.345]{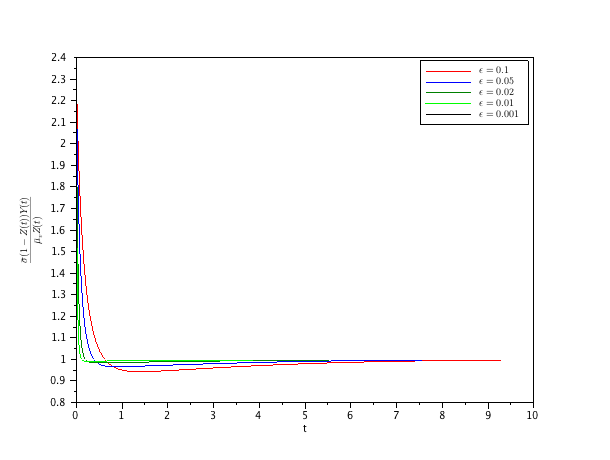}}
\subfloat[Parameter set 1 --- detail]{\includegraphics[scale=0.345]{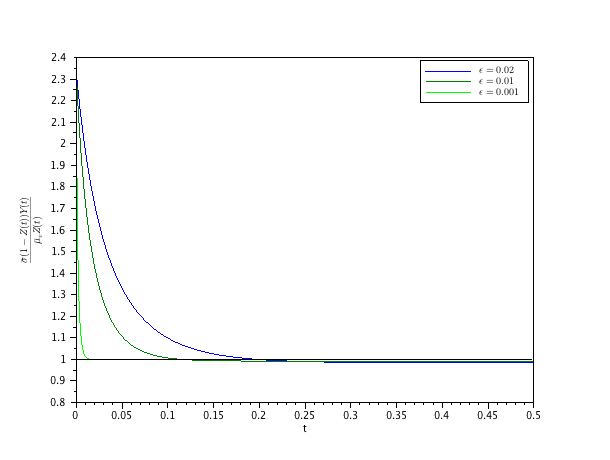}}\\
\subfloat[Parameter set 2]{\includegraphics[scale=0.345]{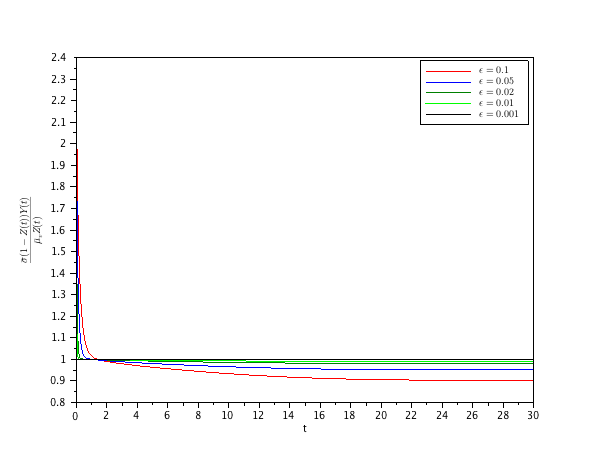}}
\subfloat[Parameter set 2 --- detail]{\includegraphics[scale=0.345]{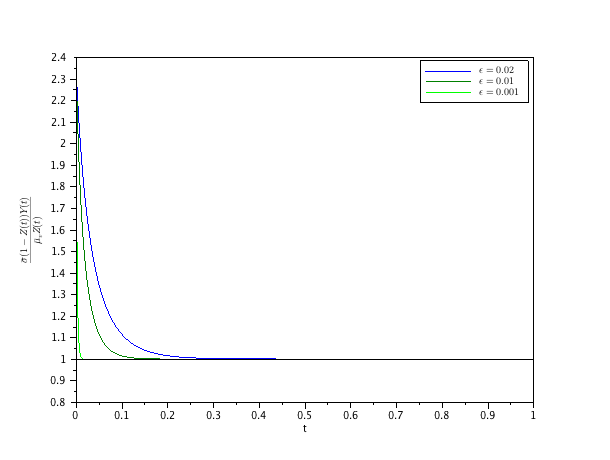}}
\caption{The dependence of the ratio given by \eqref{eqn:ratio_fvd} for a fixed initial condition as a function of the gauge parameter $\epsilon$.   For the two parameter sets, the results are similar in the sense that for $\epsilon\leq 0.02$, the approximation seems to perform very efficiently. \label{figratio_eps}.}
\end{center}
\end{figure}

\newpage

We now present further results for a third parameter set in Figure~\ref{fig:vnr_fvd}. There we  compare the full model with the asymptotic model. As expected, the approximation of $X(t)$ by $X^0(t)$ and of $Y(t)$ by $Y^0(t)$ are indeed uniform for all time, while the approximation of $Z(t)$ by $Z^0(t)$ fail to be uniform in an initial layer. Notice also that such non-uniform behaviour is suppressed by including the corrector term in the initial layer. This parameter set has $\mu_h=0.005$, $\gamma=0.4$, $\delta=0.4$, $\bar{\mu}_v=0.2$, and $\bar{\sigma}=0.5$.

\begin{figure}[htbp]
\begin{center}
\subfloat[]{\includegraphics[scale=0.345]{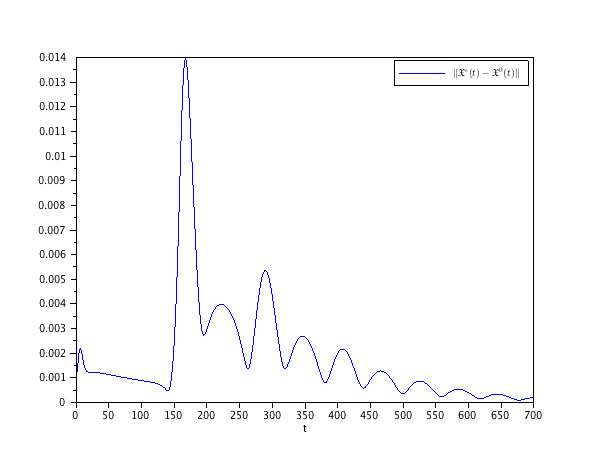}}
\subfloat[]{\includegraphics[scale=0.345]{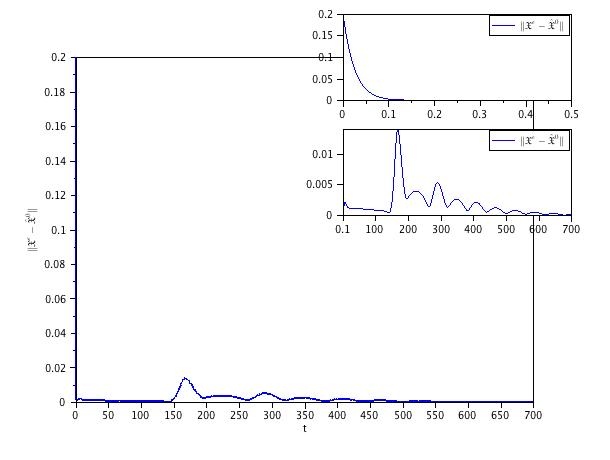}}\\
\subfloat[]{\includegraphics[scale=0.345]{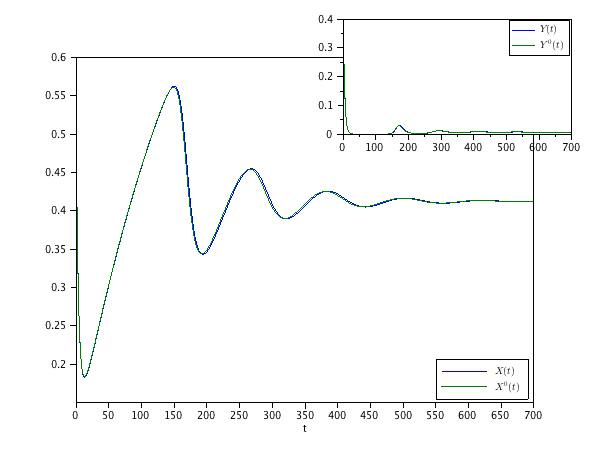}}
\subfloat[]{\includegraphics[scale=0.345]{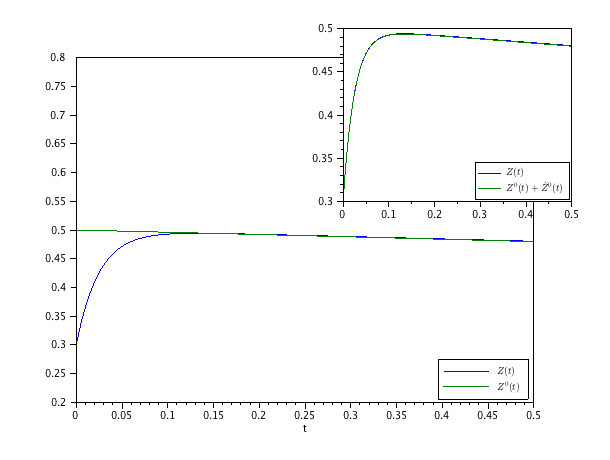}}
\caption{In (a) we show that the error in the composite approximation stays well within the expected bounds. The non uniformity within an initial transition layer is first seen in (b) when compare the norms of the leading-order  terms without the transition corrector in $Z$. The insets show that larger discrepancy indeed comes from the lack of the corrector $\hZ$. Such non uniformity is also clearly depicted in (d), with the inset showing the efectiveness of the composite approximation for the $Z$ component, while (c) shows that the leading order for $X$ (and $Y$ in the inset)---without any correctors---work also effectively well. \label{fig:vnr_fvd}}
\end{center}
\end{figure}

\begin{figure}[htbp]
\begin{center}
\includegraphics[scale=0.5]{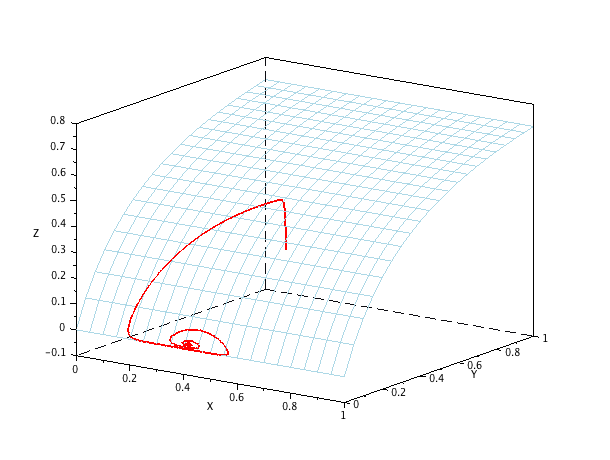}
\caption{The slow manifold dynamics. As already indicated by the Figures~\ref{figratio_ic} and \ref{figratio_eps}, we also have  that the most of the trajectory remains asymptotically close to the manifold. Parameter values are as in figure~\ref{fig:vnr_fvd}}
\label{fig:slowm_fvd}
\end{center}
\end{figure}

\newpage

\section{The fast host dynamics}
\label{sec:fhd}

Analogous to the previous case, we now assume that the dynamics of the host is much faster than the dynamics of the vector. In this case, one might expect that the  host population is nearly in equilibrium. Hence, we should have $\dot{X}\approx0$ and $\dot{Y}\approx0$, i.e, we should have the system
\begin{equation}
 \left\{ 
\begin{array}{rcl}
 0 &=& \mu_h(1-X)-\delta \frac{\sigma XY}{\sigma Y+\mu_v}\\
 0  &=& \delta \frac{\sigma XY}{\sigma Y+\mu_v} - (\mu_h+\gamma)Y\\
\dot{Z}  &=&  \sigma(1-Z)Y-\mu_v Z
\end{array}
\right.
\label{fhd:heu}
\end{equation}
Equation \eqref{fhd:heu} is a SI system for the vector, with a modified incidence rate.

In order to justify system~\eqref{fhd:heu}, we assume that
\[
 \delta=\bd\eps^{-1},\quad \mu_h=\bmh\eps^{-1}\quad\text{and}\quad\gamma=\bg\eps^{-1}.
\]
Thus, we are interested in solving
\begin{equation}
 \left\{ 
\begin{array}{rcl}
 \eps\dot{X} &=& \bmh(1-X)-\bd XZ\\
\eps\dot{Y}  &=& \bd XZ - (\bmh+\bg)Y\\
\dot{Z}  &=&  \sigma(1-Z)Y-\mu_v Z
\end{array}
\right.
\label{fhd:eps}
\end{equation}
subject to the initial condition 
\[
 X(0)=X_0,\quad Y(0)=Y_0\quad\text{and}\quad Z(0)=Z_0.
\]
In what follows, we formally derive the leading order asymptotic expansion and provide a global analysis together with some numerical results. The proof that this expansion is asymptotic is very similar to the Fast Vector Dynamics regime, and hence it is omitted.

\subsection{Asymptotic expansion}

As before, we let $\eps\tau=t$. The asymptotic expansion now take the following form: 
\begin{align*}
 X(t)&=X^0(t)+\hat{X}^0(\tau)+\O(\eps),\\
Y(t)&=Y^0(t)+\hat{Y}^0(\tau)+\O(\eps),\\
Z(t)&=Z^0(t)+\O(\eps).
\end{align*}
Here, we  also have
\[
 \lim_{\tau\to\infty}(\hat{X}^0(\tau),\hat{Y}^0(\tau)=0.
\]

At leading order, we have
\begin{align*}
 0&=\bmh(1-X^0)-\bd X^0Z^0\\
0&=\bd X^0Z^0 - (\bmh+\bg)Y^0\\
Z^0_t&=\sigma(1-Z^0)Y^0-\mu_v Z^0
\end{align*}
We solve for $X^0$ and $Y^0$ obtaining
\[
 X^0=\frac{\bmh}{\bd Z +\bmh}
\quad\text{and}\quad
Y^0=\frac{\bd\bmh}{\bmh+\bg}\frac{Z^0}{\bd Z^0+\bmh}.
\]
Thus the last equation becomes
\[
 Z^0_t=\mu_v Z^0\left(R_0D_0\frac{1-Z^0}{Z^0+D_0}-1\right).
\]
Also, we write
\[
 Z^0(t)=Z^0(\eps\tau)=Z^0(0)+\eps\tau Z^0(0)+\O(\eps^2\tau^2),
\]
hence, since $Z^0(0)=Z_0$, we obtain
\begin{align*}
 \hat{X}^0_\tau&=\bmh\hat{X}^0-\bd\hat{X}^0Z_0\\
\hat{Y}^0_\tau&=\bd\hat{X}^0Z_0=(\bg+\bmh)\hat{Y}^0
\end{align*}
We write the solution as
\[
 \begin{pmatrix}
 \hat{X}^0\\
\hat{Y}^0
\end{pmatrix}
=
\e^{tA}
\begin{pmatrix}
X_0-\frac{\bmh}{\bd Z_0+\bmh} \\
Y_0-\frac{\bd\bmh}{\bmh+\bg}\frac{Z_0}{\bd Z_0+\bmh}
\end{pmatrix},
\qquad
A=
\begin{pmatrix}
 -\bmh&-\bd Z_0\\
\bd Z_0&-(\bg+\bmh)
\end{pmatrix}
.
\]
It is straightforward to verify  that the eigenvalues of $A$ always have negative real part, and hence that
\[
 \lim_{\tau\to\infty}(X^0(\tau),Y^0(\tau))=0.
\]

\subsection{Global stability analysis}

The equilibria are $Z_0=0$ and $Z_0=Z^*$. Since
\[
 \frac{\rd }{\rd Z_0}\left(Z_0\left[R_0D_0\frac{1-Z_0}{D_0+Z_0} - 1\right]\right)=
R_0D_0\frac{1-Z_0}{D_0+Z_0}-1-R_0D_0\frac{1+D_0}{(D_0+Z_0)^2}Z_0
\]
At $Z_0=0$, its value is $R_0-1$. So the origin is globally asymptotically stable for $R_0<1$.

At $Z_0=Z^*$ its value is
\[
 \frac{1+R_0D_0}{R_0(1+D_0)}(1-R_0)
\]
Hence $Z^*$ is globally asymptotically stable, if $R_0>1$.

When $R_0=1$, we have
\[
 Z_0\left[R_0D_0\frac{1-Z_0}{D_0+Z_0} - 1\right]=-\frac{D_0Z_0}{D_0+Z_0}(Z_0+1)<0, \quad Z_0\geq0.
\]
Hence, $Z_0=0$ is also globally asymptotic stable when $R_0=1$.

\subsection{Numerical results}

The results for the components are qualitatively similar to the fast vector dynamics, and hence are omitted. Nevertheless, the reduction of the dynamics to the slow manifold is more dramatic in this case as  shown in figure~\ref{fig5}.

\begin{figure}[htbp]
 \includegraphics[scale=0.5]{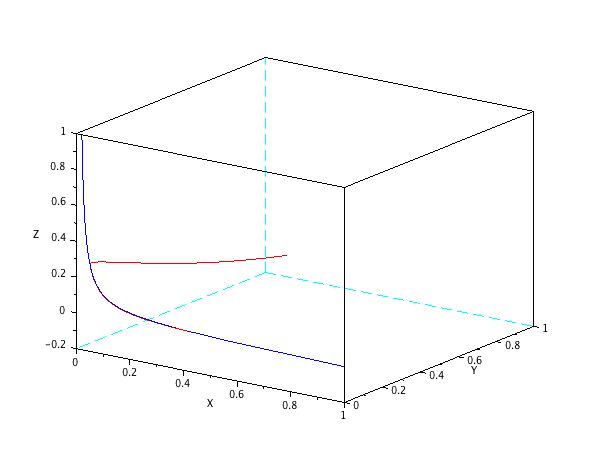}
\caption{Slow manifold dynamics for the fast host regime. Parameter values are $\bmh=0.005$, $\bg=0.4$, $\bd=0.4$, $\mu_v=0.2$, and $\sigma=0.5$. }
\label{fig5}
\end{figure}

\newpage

\section{Concluding remarks}
\label{sec:conclusion}

As observed in the introduction, diseases that are vector-borne  have a number of features that distinguish them from contagious ones. Typically, time scales for the dynamics of the host and vectors are not within the same order, since mosquitoes, for instance, that are a prevalent vector for such diseases have a very fast life cycle compared to humans. With this in mind, we investigated the dynamical consequences of  having host and vector dynamics with distinct time scales in the classical arbovirus model introduced by \cite{Bailey75,Dietz75}. The natural regimes to look in this model are the \textit{fast vector dynamics} (FVD) and \textit{fast host dynamics} (FHD). While the former seems to be the most natural choice, we take the view that there might be scenarios where the latter may be observed.

By means of a formal multiscale asymptotic analysis, we study both regimes. For the FVD, we find the leading order dynamics yields a SIR model for the host, with a modified incidence rate. Thus, the vector is removed from the model being present only parametrically as a function of the host infected fraction. Such a relationship, apart from its mathematical interest, might also be useful in verifying if field data conforms, within the model, with the regime hypothesis.   Additionally, the FHD regime yields an even more dramatic reduction with an SI model for the vectors, again with a modified incidence rate. Numerical results presented show that the approximation is indeed uniformly asymptotic in time. 
An interesting feature of the studied regimes is that they do not imply any condition on $R_0$, and hence are compatible with a variety of disease developments from the point of view of global dynamics. Indeed, for both reduced models, the equilibria are  preserved by the asymptotic approximation, and the global stability dynamics is consistent with the global stability dynamics of the full model.  Finally, we have confirmed rigorously the asymptotic character of the approximation up to the derived order.
Due to the large variance of the parameters as measured in many different cases, we do not claim that these regimes are necessarily the most important,  or the more prevalent. Nevertheless, they do provide a model problem where the reduction can be effectively carried out,  and indicate how the dynamics can be described by such reduced models. In addition, the parameters  that give rise to such regimes, particularly the Fast Vector Dynamics, are consistent with a number of concrete epidemiological scenarios.

The results obtained suggest that  multiscaling reductions similar to the ones described here might be very effective in obtaining simplified models. In particular, as the complexity of  models grows, we expect that such reductions may provide additional insights.

\appendix

\newpage

\section{Proof of Theorem~\ref{asymp_thm}}
\label{app:asymp_proof}

The proof of Theorem~\ref{asymp_thm} is divided in several lemmas.

We write system \eqref{fvd:eps} in a more concise form as
\begin{align*}
 \dot{\bW}&=\mathcal{F}(\bW,Z),\\
\eps\dot{Z}&=\mathcal{G}(\bW,Z).
\end{align*}
with $\bW(t)=(X(t),Y(t))^t$ and $\mathcal{F}$ and $\mathcal{G}$ being the appropriate entries of the right hand
side of \eqref{fvd:eps}.
We  write
\[
 \bW=\bW^0+\eps\hbW^0+\eps\bQ
\quad\text{and}\quad
Z=Z^0+\hZ^0+\eps \bZ,
\]
with
\[
 \hbW^0=\frac{X_0}{\bs Y_0+\bmv}\left(Z_0-\frac{\bs Y_0}{\bs Y_0+\bmv}\right)\e^{-(\bs Y_0 +\bmv)t/\eps}(1,-1)^t.
\]
Notice that since $\hbW^0$ is bounded, we need only to prove that $(\bQ,\bZ)$ exist, are bounded, and are unique. In this case, we then take $C=\|(\bQ,\bZ)\|_\infty$.
First, we observe that
\[
 \dot{\bW^0}+\eps\dot{\hbW^0}=\mathcal{F}(\bW^0,Z^0+\hZ^0)+K(t)(1,-1)^t
\quad\text{and}\quad
\dot{\hZ^0}=\mathcal{G}(\bW^0,Z^0+\hZ^0)+L(t),
\]
where
\[
 K(t)=\delta\hZ^0(t/\eps)(X^0(t)-X_0)
\quad\text{and}\quad
L(t)=\bs\hZ^0(t/\eps)(Y^0(t)-Y_0).
\]
In particular, because of the fast decay of $\hZ^0$, and because $L(0)=K(0)=0$, it follows that there exists a constant $C>0$, such that
\[
 \int_0^\infty K(t)\,\rd t,\quad \int_0^\infty L(t)\,\rd t\leq C\eps^2.
\]
Since $\mathcal{F}$ and $\mathcal{G}$ are quadratic, we write:
\begin{align*}
 \mathcal{F}(\bW,Z)=&\mathcal{F}(\bW^0,Z^0+\hZ^0)+\eps D_{\bW}\mathcal{F}(\bW^0,Z^0+\hZ^0)(\bQ +\hbW^0)+ 
\eps D_{Z}\mathcal{F}(\bW^0,Z^0+\hZ^0)\bZ +\\
&\qquad  +\eps^2\delta\bX\bZ\begin{pmatrix}
 -1\\1
\end{pmatrix};
\\
\mathcal{G}(\bW,Z)=&\mathcal{G}(\bW^0,Z^0+\hZ^0)+\eps D_{\bW}\mathcal{G}(\bW^0,Z^0+\hZ^0)(\bQ +\hbW^0)+ 
\eps D_{Z}\mathcal{G}(\bW^0,Z^0+\hZ^0)\bZ +\\
& +\qquad\eps^2\bs\bY \bZ.
\end{align*}
where $\bQ=(\bX,\bY)^t$.

Let $T(t,s)$ be the fundamental solution to the linearised system
\begin{equation}
 \label{lin_sys}
\begin{pmatrix}
 \dot{\bQ}\\\eps\dot{\bZ}
\end{pmatrix}
=\begin{pmatrix}
D_{\bW}\mathcal{F}(\bW^0,Z^0+\hZ^0)&D_{Z}\mathcal{F}(\bW^0,Z^0+\hZ^0)\\
D_{\bW}\mathcal{G}(\bW^0,Z^0+\hZ^0)&D_{Z}\mathcal{G}(\bW^0,Z^0+\hZ^0)
\end{pmatrix}
\begin{pmatrix}
 \bQ\\\bZ
\end{pmatrix}
\end{equation}
Then direct integration yields
\begin{lem}
\label{lem:iub}
The functions $(\bQ,\bZ)$ satisfy the following integral equation:
\begin{align}
 \begin{pmatrix}
 \bQ\\\bZ
\end{pmatrix}
=&T(t,0)\begin{pmatrix}
 Q_0\\\bZ_0
\end{pmatrix}
+
\int_0^tT(t,s)\begin{pmatrix}
\epsilon\delta\bX\bZ\begin{pmatrix}
 1\\-1
\end{pmatrix}\\
\bs\bY \bZ
\end{pmatrix}\,\rd s + \nonumber\\
&\qquad
+\int_0^tT(t,s)\begin{pmatrix}
D_W\mathcal{F} \cdot\hbW^0(s)\\ D_W\mathcal{G} \cdot\hbW^0(s)-\dot{Z}_0(s)
\end{pmatrix}\,\rd s
-
\frac{1}{\eps^2}\int_0^tT(t,s)\begin{pmatrix}
 \eps K(s)\begin{pmatrix}
 1\\-1
\end{pmatrix}
\\ L(s)
\end{pmatrix}\,\rd s.
\label{iasymp}
\end{align}
Moreover, the last term is bounded uniformly in $\eps$.
\end{lem}
We also have the following large time behaviour result for the linearised system \eqref{lin_sys}:
\begin{lem}
\label{lem:ltb}
 Let $(\bQ(t),\bZ(t))$ be a solution to \eqref{lin_sys}. Then
\[
 \lim_{t\to\infty}(\bQ(t),\bZ(t))=\mathbf{0}.
\]
In particular, the solutions to \eqref{lin_sys} are bounded uniformly in time for any given $\eps$. Moreover, they are also uniformly bounded in $\eps\leq1$, for all $t\geq0$.
\end{lem}

\begin{proof}
For notation convenience, let us write \eqref{lin_sys} as
\[
 \begin{pmatrix}
 \dot{\bQ}\\\eps\dot{\bZ}
\end{pmatrix}
= A(\bW^0,Z^0+\hZ^0)
\begin{pmatrix}
 \bQ\\\bZ
\end{pmatrix}.
\]
Fix $\eps>0$ and $(X_0,Y_0,Z_0)$. From Theorem~\ref{thm:pl_gs}, we know that
\[
 \lim_{t\to\infty} \left(\bW^0(t),Z^0(t)+\hZ^0(t/\eps)\right)=(\bW^*,Z^*),
\]
where $(\bW^*,Z^*)$ is the globally asymptotic stable equilibrium given by Theorem~\ref{thm:3d_gs}. Therefore, there exists $T>0$, such that $t>T$ implies that $A$ is negative-definite. Since
$\hZ^0\to0$, as $\eps\to0$. We can choose $T$ such this holds for all $0<\eps\leq1$. 

Because of the continuity of the solution with respect to the initial conditions, we have that $T$ is a continuous function of the initial conditions. Since these lie on a compact set, we can pick $T$ such that $A$ is negative definite for all $0<\eps\leq1$ and for all $(X_0,Y_0,Z_0)$, such that $X_0,Y_0\geq0$, $X_0+Y_0\leq1$ and $0<Z_0\leq1$. But then, for any such $(X_0,Y_0,Z_0)$ and $0<\eps\leq1$, we have that 
\[
 \lim_{t\to\infty}T(t,T)=0.
\]
Therefore, for any initial condition $(Q_0,\bZ_0)^t$, we have
\[
 \lim_{t\to\infty}T(t,0)\begin{pmatrix}Q_0\\\bZ_0\end{pmatrix}=
\lim_{t\to\infty}T(t,T)T(T,0)\begin{pmatrix}Q_0\\\bZ_0\end{pmatrix}=0.
\]
\qed
 \end{proof}

\begin{proof}[Proof of Theorem~\ref{asymp_thm}]

First, we observe that the nonlinear term in \eqref{iasymp} is locally Lipschitz, hence a standard fixed point  yields existence and uniqueness for $0\leq t <t_0$ , for some, possibly small, $t_0$.

  From the decaying of $\hbW^0$ and from  Lemma~\ref{lem:iub}, we conclude that the two last terms of \eqref{iasymp} are uniformly bounded in time and $\eps$. 

Moreover, Lemma~\ref{lem:ltb} implies that the first two terms on the right hand side of \eqref{iasymp} are also uniformly bounded in time and in $\eps$, for sufficiently small $\eps$. Thus, we conclude that the same also holds true for  $(Q,R)^t$.
Therefore, the solution to \eqref{iasymp} is globally defined in time, and it is bounded uniformly in $\eps$, if the latter is sufficiently small. 

\qed
\end{proof}

\bibliographystyle{abbrv}
\bibliography{dengue}

\end{document}